\documentclass[11pt]{article}

\usepackage{amsfonts}
\usepackage{amsmath}
\usepackage{amssymb}
\usepackage{amsthm}
\usepackage{mathrsfs}
\usepackage{graphicx}

\usepackage[]{algorithm2e}

\theoremstyle{plain}
\newtheorem{theorem}{Theorem}
\newtheorem{proposition}[theorem]{Proposition}

\newtheorem{corollary}[theorem]{Corollary}

\theoremstyle{definition}

\newtheorem{remark}[theorem]{Remark}

\numberwithin{exercise}{section}
\numberwithin{equation}{section}
\numberwithin{theorem}{section}
\numberwithin{problem}{section}

\numberwithin{figure}{section}

\newcommand{\bs}[1]{{\boldsymbol{#1}}}

\newcommand{\D}{\,\mathrm{d}}

\begin{document}

\title{How trait distributions evolve in heterogeneous populations}

\author{Georgy P. Karev$^{1,} $\footnote{e-mail: karev@ncbi.nlm.nih.gov}\,\,, Artem S. Novozhilov$^{2,}$\footnote{Corresponding author, e-mail: artem.novozhilov@ndus.edu}\\[3mm]
\textit{\normalsize $^\textrm{\emph{1}}$National Center for Biotechnology Information, National Library of Medicine,}\\[-1mm]
\textit{ National Institutes of Health, Bethesda, MD, USA}\\[1mm]
\textit{\normalsize $^\textrm{\emph{2}}$Department of Mathematics, North Dakota State University, Fargo, ND 58108, USA}}
\date{}

\maketitle

\begin{abstract}We consider the problem of determining the time evolution of a trait distribution in a mathematical model of non-uniform populations with parametric heterogeneity. This means that we consider only heterogeneous populations in which heterogeneity is described by an individual specific parameter that differs in general from individual to individual, but does not change with time for the whole lifespan of this individual. Such a restriction allows obtaining a number of simple and yet important analytical results. In particular we show that initial assumptions on time-dependent behavior of various characteristics, such as the mean, variance, of coefficient of variation, restrict severely possible choices for the exact form of the trait distribution. We illustrate our findings by in-depth analysis of the variance evolution. We also reanalyze a well known mathematical model for gypsy moth population showing that the knowledge of how distributions evolve allows producing oscillatory behaviors for highly heterogeneous populations.

\paragraph{\small Keywords:} Heterogeneous populations, parametric heterogeneity, evolution of distribution, normal distribution, gamma distribution


\end{abstract}

\section{Introduction}
The real world populations are inherently heterogeneous, with individuals differ by age, spatial location, contact network, social habits, genome composition, etc. When various ecological and evolutionary processes are expected to be modeled with the help of mathematical models, this heterogeneous structure of populations should be included into the models. In rather general settings this assumption leads to the so-called structured population models (e.g., \cite{metz2014dynamics,perthame2006transport}), whose mathematical analysis, however, is quite demanding. In certain cases significant simplifications are possible, in particular when one speaks of the so-called \textit{parameteric heterogeneity} when the structure of a population is described by the variation in population parameters (such as birth rate or susceptibility to a particular disease); the parameters are considered as invariant properties of individuals that do not change their values throughout the lifespan of an individual, whereas they vary from individual to individual. As an example of an arguably simplest mathematical model with parametric heterogeneity consider the following Malthusian model
\begin{equation}\label{eq1:1}
    \frac{\partial n}{\partial t}(t,\omega)=\omega n(t,\omega),
\end{equation}
where $n(t,\omega)$ is the density of the population having specific parameter value $\omega$ at time $t$, and $\omega\in\Omega$ is the distributed parameter, $\Omega$ is a set of admissible values of $\omega$. To solve problem \eqref{eq1:1} one needs an initial condition of the form
\begin{equation}\label{eq1:2}
    n(0,\omega)=n_0(\omega)=N_0p(0,\omega),
\end{equation}
where $N_0$ is the initial size of the population and $p(0,\omega)$ is the given initial probability distribution function (pdf) of the parameter distribution.

Problem \eqref{eq1:1}, \eqref{eq1:2} provides a simplest example of a mathematical model with parametric heterogeneity, and at the same time is a very flexible object (see, e.g., \cite{Karev2000a,Karev2003,Karev2005}). It can be shown (see the cited references) that the total population size $N(t)=\int_\Omega n(t,\omega)\D \omega$ is explicitly given by
\begin{equation}\label{eq1:3}
    N(t)=N_0 \textsf{M}(0,t),
\end{equation}
where $\textsf{M}(0,\lambda)$ is the initial moment generating function (mgf) of the pdf $p(0,\omega)$:
\begin{equation}\label{eq1:4}
    \textsf{M}(0,\lambda)=\int_{\Omega}e^{\omega \lambda}p(0,\omega)\D\omega.
\end{equation}

It turns out that in much more general situations the initial mgf $\textsf{M}$ is the key device allowing to untangle the system's dynamics. To be precise, in what follows we consider the population models with parametric heterogeneity in the following form (we note that more general situations can be similarly treated, see, e.g., \cite{karev2009})
\begin{equation}\label{eq1:5}
    \frac{\partial n}{\partial t}(t,\omega)=n(t,\omega)\bigl(f(\bs v)+\omega g(\bs v)\bigr),
\end{equation}
where $f,g$ are two given functions that depend on the vector $\bs v=(t,N(t),\textsf{E}_t[\omega],\ldots)$, e.g., they may depend on the aggregated characteristics of $n(t,\omega)$ such as the total population size $N(t)$ at time $t$, the mean of the parameter $\textsf{E}_t[\omega]=\int_\Omega \omega p(t,\omega)\D \omega$, on other moments, but not on the density $n(t,\omega)$ or parameter $\omega$ themselves. The initial condition for model \eqref{eq1:5} is  given by \eqref{eq1:2}. We remark that model \eqref{eq1:5} is readily generalized to vector-parameters $\bs \omega$ and systems of equations. Despite its quite special form, model \eqref{eq1:5} was used for modeling various biological processes (e.g., \cite{Ackleh1999,boylan1991note,coutinho1999modelling,Dwyer1997,Karev2000,Karev2005,Karev2003,karev2010,karev2006,novozhilov2004,Novozhilov2012,Novozhilov2008,Veliov2005}). It can be shown (see, e.g., \cite{karev2009}) that model \eqref{eq1:5}, \eqref{eq1:2} is equivalent to a system of ordinary differential equations (ODE) of the form
\begin{equation}\label{eq1:6}
    \begin{split}
      \frac{\D N}{\D t}(t) &= N(t)\bigl(f(\bs v)+\textsf{E}_t[\omega]g(\bs v)\bigr), \\
       \textsf{E}_t[\omega] & =\left.\frac{\D}{\D \lambda}\textsf{M}(0,\lambda)\right|_{\lambda=q(t)}\cdot\frac{1}{\textsf{M}(0,q(t))}\,,\\
       \frac{\D q}{\D t}(t)&=g(\bs v),\\
       q(0)&=0,\quad N(0)=N_0.
    \end{split}
\end{equation}

The key fact for the present note is that the mgf $\textsf{M}(t,\lambda)$ of the parameter $\omega$ at any time moment $t$ can be expressed through the initial mgf $\textsf{M}(0,\lambda)$ and an auxiliary variable $q(t)$:
\begin{equation}\label{eq1:7}
    \textsf{M}(t,\lambda)=\frac{\textsf{M}(0,\lambda+q(t))}{\textsf{M}(0,q(t))}\,,
\end{equation}
which in principle allows computing any moments of the parameter distribution at any time.

We remark that a similar relation can be obtained within the frailty model (see, e.g., \cite{aalen1994effects,aalen2008survival}) in statistical analysis.

An immediate consequence of \eqref{eq1:7} is that model \eqref{eq1:5} severely restricts possible time dependent evolution of the parameter distribution. In particular, we have, as an immediate consequence of relation \eqref{eq1:7} the following
\begin{proposition}\label{pr:1:1}
$(i)$ Let the initial distribution $p(0,\omega)$ in the problem \eqref{eq1:5}, \eqref{eq1:2} be a gamma-distribution with parameters $k,\nu$:
$$
p(0,\omega)=\frac{\nu^k}{\Gamma(k)}w^{k-1}e^{-\nu \omega},\quad \omega\geq 0,\,k>0,\,\nu>0.
$$
Then at any time moment $t$ (for which the model is defined) the parameter distribution is again a gamma-distribution, with parameters $k$ and $\nu-q(t)$. In particular, the coefficient of variation $\textsf{CV}$ is constant:
$$
\textsf{\emph{CV}}=\frac{\sqrt{\textsf{\emph{Var}}_t[\omega]}}{\textsf{\emph{E}}_t[\omega]}=\frac{1}{\sqrt{k}}\,.
$$

$(ii)$ Let the initial distribution $p(0,\omega)$ in the problem \eqref{eq1:5}, \eqref{eq1:2} be a normal distribution with parameters $\mu,\sigma$:
$$
p(0,\omega)=(\sqrt{2\pi}\sigma)^{-1}\exp\left\{-\frac{(\omega-\mu)^2}{2\sigma^2}\right\},\quad \sigma>0.
$$

Then at any time moment $t$ the parameter distribution is again a normal distribution with parameters $\mu+2 q(t)\sigma^2,\sigma$. In particular, the variance of the distribution is constant.
\end{proposition}

Proposition \ref{pr:1:1} can easily be extended to various other initial distributions. We would like to mention only few additional examples that are important from a practical point of view and which we use in the current text to highlight and illustrate general results.
\begin{proposition}\label{pr:1:2}
$(i)$ Let the initial distribution in the problem \eqref{eq1:5}, \eqref{eq1:2} be Poissonian with the mean $\mu$:
$$
\textsf{\emph{Pr}}_0(\omega=i)=\frac{\mu^i}{i!}e^{-\mu},\quad i=0,1,\ldots.
$$
Then at any time moment $t$ the parameter distribution is again Poissonian with the mean $\textsf{\emph{E}}_t[\omega]=\mu e^{q(t)}$.

$(ii)$ Let the initial distribution be truncated exponential, which means that
$$
p(0,\omega)=C e^{-s \omega},
$$
where $0<\omega<b$ and $C=s(1-e^{-sb})^{-1}$ is the normalization constant. Then at any time moment $t$ the distribution is also truncated exponential in the same interval $[0,b]$ with the parameter $s-q(t)$.
\end{proposition}

In both Propositions \ref{pr:1:1} and \ref{pr:1:2} the form of the distribution does not change with time, only its parameters evolve. This is not true for an arbitrary distribution. For instance, an immediate consequence of \eqref{eq1:7} is
\begin{proposition}\label{pr:1:3}
Let the initial distribution in the problem \eqref{eq1:5}, \eqref{eq1:2} be uniform in the interval $[a,b]$. Then at any time moment $t$ the parameter distribution is truncated exponential in the interval $[a,b]$ and has the form
$$
p(t,\omega)=q(t)\frac{\exp(q(t)\omega)}{\exp(q(t)b)-\exp(q(t)a)}\,.
$$
\end{proposition}
\section{Inverse problem}

Now, since we know that the form of the mathematical model \eqref{eq1:5} puts some restrictions on the evolution of the trait distributions, and in some cases these restrictions imply conservations of particular quantities (variance in the case of the normal distribution, coefficient of variation in the case of the gamma distribution, etc), it is natural and important to ask the inverse question: What if we know that, for instance, the coefficient of variation is conserved with time. Does it imply that the initial distribution must be a gamma distribution?

To answer this question consider again the key relation \eqref{eq1:7}. From this relation we obtain that
\begin{equation}\label{eq2:0}
\begin{split}
\textsf{E}_t[\omega]&=\left.\frac{\D }{\D \lambda}\textsf{M}(t,\lambda)\right|_{\lambda=0}=\frac{\frac{\D}{\D \lambda} \textsf{M}(0,\lambda)|_{\lambda=q(t)}}{\textsf M(0,q(t))}\,,\\
\textsf{Var}_t[\omega]&=\left.\frac{\D^2 }{\D \lambda^2}\textsf{M}(t,\lambda)\right|_{\lambda=0}-\left(\textsf{E}_t[\omega]\right)^2=\frac{\frac{\D^2}{\D \lambda^2} \textsf{M}(0,\lambda)|_{\lambda=q(t)}}{\textsf M(0,q(t))}-\left(\frac{\frac{\D}{\D \lambda} \textsf{M}(0,\lambda)|_{\lambda=q(t)}}{\textsf M(0,q(t))}\right)^2\,,
\end{split}
\end{equation}
and so on.

To simplify the notation, let us introduce the function $m(q)=\textsf{M}(0,q)$ that depends on the auxiliary variable $q$, which we consider to be an \textit{internal time} of our systems (the internal time for the models in the form \eqref{eq1:5} was introduced in \cite{karev2016mathematical}, where an extended discussion of this concept is given).

First we consider the dynamics of the mean $ \textsf{E}_q[\omega]=\mu(q)$, which, according to the first equation in \eqref{eq2:0}, turns into the first order ODE
\begin{equation}\label{eq:2:a}
\frac{m'(q)}{m(q)}=\mu(q)
\end{equation}
with the initial condition $m(0)=1$. In particular, we immediately obtain that if the dynamics of the mean is given by a linear function of $q$: $\mu(q)=2\sigma^2 q+\mu$ then the initial distribution must be normal, because in this case we have $m(q)=e^{2\sigma^2 q+\mu}$, which is the mgf of the normal distribution. Similarly, assuming that $\mu(q)=\mu e^{q}$ implies the initial Poissonian distribution with the mgf $\textsf{M}_0(q)=e^{\mu(\exp(q)-1)}$. Furthermore, assuming that $\mu(q)=\eta+k(\nu-q)^{-1}$ implies $m(q)=\exp(\eta q)\left(1-\frac{q}{\nu}\right)^{-k}$ and hence initial gamma distribution with parameters $\eta,\nu,k$; for the special case $\eta=0,k=1$ we get the initial exponential distribution with mean $1/\nu$. The list of possible distributions can be readily extended.

As a final remark about the ODE \eqref{eq:2:a}, we note that $\mu(q)=\mbox{const}=E$ implies $m(q)=e^{E q}$, which is formally the mgf of the delta-function concentrated at the point $\omega=E$. This implies that within the framework of model \eqref{eq1:5} assumption on the constant mean of the parameter distribution is equivalent to the assumption that the population is homogeneous.

Now we turn to the second equation in \eqref{eq2:0}. Te expression for the variance $\sigma^2(q)=\textsf{Var}_q[\omega]$ can be written as
\begin{equation}\label{eq2:1}
\frac{m''(q)}{m(q)}-\left(\frac{m'(q)}{m(q)}\right)^2=\sigma^2(q),
\end{equation}
where the derivatives are taken with respect to $q$. Equation \eqref{eq2:1} is an ODE with the natural initial conditions
\begin{equation}\label{eq2:2}
    m(0)=1,\,m'(0)=\textsf{E}_t[\omega]=\mu,
\end{equation}
which can be readily solved as
$$
m(q)=\exp\left\{\int_0^q\int_0^s\sigma^2(\tau)\D \tau\D s+\mu q\right\}\,.
$$
The last expression gives us a family of mgf depending on the behavior of the variance of the distribution of the parameter at the internal time moment $q$. In particular, if one assumes that $\sigma^2(q)=\sigma^2$ is a constant then
$$
m(q)=\exp\left\{\sigma^2\frac{q^2}{2}+\mu q\right\},
$$
which is exactly the mgf of the normal distribution.

Similarly, one can assume that in equation \eqref{eq2:1} the coefficient of variation is conserved, and hence the ODE takes the form
\begin{equation}\label{eq2:3}
\frac{m''(q)}{m(q)}-\left(\frac{m'(q)}{m(q)}\right)^2=\left(\frac{m'(q)}{m(q)}\right)^2 \textsf{CV}^2,
\end{equation}
or, in terms of the parameters $k,\nu$
$$
\frac{m''(q)}{m(q)}-\left(\frac{m'(q)}{m(q)}\right)^2=\left(\frac{m'(q)}{m(q)}\right)^2\frac{1}{k}\,\quad m(0)=1,\,m'(0)=\frac{k}{\nu}\,.
$$
The solution to the last initial value problem is
$$
m(q)=\left(1-\frac{q}{\nu}\right)^{-k}\,,
$$
which is the mgf of the gamma distribution with parameters $k,\nu$. Therefore we have proved
\begin{proposition}\label{pr:2:1}
Let \eqref{eq1:5} be the mathematical model of the dynamics of a heterogeneous population and assume that the variance of the parameter distribution does not change with time. Then the initial (and at any time moment $t$) distribution of the parameter must be normal.

Let \eqref{eq1:5} be the mathematical model of the dynamics of a heterogeneous population and assume that the coefficient of variation of the parameter distribution does not change with time. Then the initial (and at any time moment $t$ for which the model is defined) distribution of the parameter must be a gamma distribution.

\end{proposition}

We emphasize that taken together Propositions \ref{pr:1:1} and \ref{pr:2:1} show that within the framework of the heterogeneous populations with parametric heterogeneity the assumption on the initial distribution to be a normal one (a gamma distribution) is \textit{equivalent} to assuming that the variance (respectively, the coefficient of variation) of the given trait distribution that evolves in time is \textit{constant}.

\begin{remark}It is interesting (and intriguing) to remark that prescription of the initial variance $\sigma^2$ implies the normal distribution in the settings of the maximum entropy models \cite{kapur1989maximum}. To obtain a gamma-distribution in the same framework requires fixing both the arithmetic and geometric means for the distribution.
\end{remark}

\begin{remark}The list of ODE that produce mgf for some specific conditions can be easily extended. We included a few more examples in Appendix.
\end{remark}
\section{Dynamics of the variance}
In this section we start with a more general mathematical model of the form
\begin{equation}\label{eq3:1}
    \frac{\partial n}{\partial t}(t,\omega)=n(t,\omega)F(\omega,\bs v),
\end{equation}
where $F$ is a given function, and other notation as in \eqref{eq1:5}. The variance of the parameter distribution at any time moment is equal, by definition, to
$$
\sigma^2(t)=\textsf{Var}_t[\omega]=\textsf{E}_t\left[(\omega-\textsf{E}_t[\omega])^2\right].
$$
The differential equation for the evolution of the variance directly follows from the Price equation (e.g., \cite{page2002unifying})
$$
\frac{\D}{\D t}\textsf{E}[z_t]=\textsf{Cov}[z_t,F]+\textsf{E}\left[\frac{\D z_t}{\D t}\right],
$$
where $z_t$ is a given trait value at the time $t$. Taking $z_t=\sigma^2(t)$ we obtain
\begin{proposition}\label{pr:3:1}
Within the framework of mathematical model \eqref{eq3:1}
$$
\frac{\D \sigma^2(t)}{\D t}=\textsf{\emph{Cov}}[F,(\omega-\textsf{\emph{E}}_t[\omega])^2].
$$
\end{proposition}
\begin{proof}Indeed,
\begin{align*}
\frac{\D \sigma^2(t)}{\D t}&=\frac{\D \textsf{E}_t[(\omega-\textsf{E}_t[\omega])^2]}{\D t}=\textsf{Cov}[F,(\omega-\textsf{{E}}_t[\omega])^2]- 2 \textsf{E}_t\left[(\omega-\textsf{E}_t[\omega])\frac{\D \textsf{E}_t[\omega]}{\D t}\right].
\end{align*}
Since
$$
\textsf{E}_t\left[(\omega-\textsf{E}_t[\omega])\frac{\D \textsf{E}_t[\omega]}{\D t}\right]=\frac{\D \textsf{E}_t[\omega]}{\D t} \textsf{E}_t[\omega-\textsf{E}_t[\omega]]=0,
$$
we obtain the desired result.
\end{proof}


More can be said in the case of model \eqref{eq1:5}. Indeed, in this case we have
\begin{align*}
\textsf{Cov}[f(\bs v)+\omega g(\bs v),(\omega-\textsf{E}_t[\omega])^2]&=g(\bs v)\textsf{Cov}[\omega ,(\omega-\textsf{E}_t[\omega])^2]\\
&=g(\bs v)\left(\textsf{E}_t\left[\omega(\omega-\textsf{E}_t[\omega])^2\right]-\textsf{E}_t[\omega]\textsf{E}_t\left[(\omega-\textsf{E}_t[\omega])^2\right]\right)\\
&=g(\bs v)\left(\textsf{E}_t\left[(\omega-\textsf{E}_t[\omega])(\omega-\textsf{E}_t[\omega])^2\right]-\textsf{E}_t[\omega]\textsf{E}_t\left[(\omega-\textsf{E}_t[\omega])^2\right]\right)\\
&=g(\bs v)\left(\textsf{E}_t\left[(\omega-\textsf{E}_t[\omega])^3\right]\right),
\end{align*}
and hence we have proved
\begin{proposition}\label{pr:3:2}
Within the framework of model \eqref{eq1:5}
$$
\frac{\D \sigma^2(t)}{\D t}=g(\bs v)\mu_3(t),
$$
where $\mu_3(t)=\textsf{E}_t\left[(\omega-\textsf{E}_t[\omega])^3\right]$ is the third central moment of the parameter distribution at time $t$.
\end{proposition}

\begin{corollary}
Let the parameter distribution $p(t,\omega)$ in the model \eqref{eq1:5} be symmetric at any time $t$. Then this distribution is normal.
\end{corollary}
Indeed, since we assumed that $\mu_3(t)=0$ then by Proposition \ref{pr:3:2} the variance does not change with time, and hence by Proposition \ref{pr:2:1} the only distribution that satisfies this condition is the normal one. See also Appendix for a different proof of the same result.

We emphasize that the dynamics of the variance within the framework of considered mathematical models depends critically on the initial distribution. Let us illustrate it further using the simplest heterogeneous model \eqref{eq1:1} for which the internal time $q(t)$ coincides with the system time $t$. We already showed that the initial normal distribution implies constant in time variance. In contrast, assuming that the initial distribution is exponential with the initial mean $\textsf{E}_0[\omega]=1/T$ for some constant $T$ the key relation \eqref{eq1:7} implies that the distribution will be exponential for any time moment $t<T$ with the mean $\textsf{E}_t[\omega]=\frac{1}{T-t}$ and the variance $\sigma^2(t)=\frac{1}{(T-t)^2}$. Now we see that as $t\to T$ the variance tends to infinity.

Finally, together with the model \eqref{eq1:1} one can consider a truncated exponential distribution (see Proposition \ref{pr:1:2}) on the interval $[0,b]$ (or any distribution concentrated on a compact interval). Then according to the equation that holds for the model \eqref{eq1:1}
$$
\frac{\D \textsf{E}_t[\omega]}{\D t}=\sigma^2(t),
$$
which is a version of Fisher's fundamental theorem of natural selection (e.g., \cite{page2002unifying}), the mean value $\textsf{E}_t[\omega]$ increases until the population stops being heterogeneous; it means that the distribution in the course of time will concentrate at the point $b$ and hence the variance will vanish.

To reiterate, the dynamics of the variance (and all other statistical characteristics of the population) depend in an essential way on the initial distribution: even arbitrarily small value of the variance at the initial time moment cannot guarantee that it will not increase indefinitely, and even arbitrarily large initial variance may eventually vanish, we illustrate this fact in Figure \ref{fig:2}.
\begin{figure}[!t]
\centering
\includegraphics[width=0.5\textwidth]{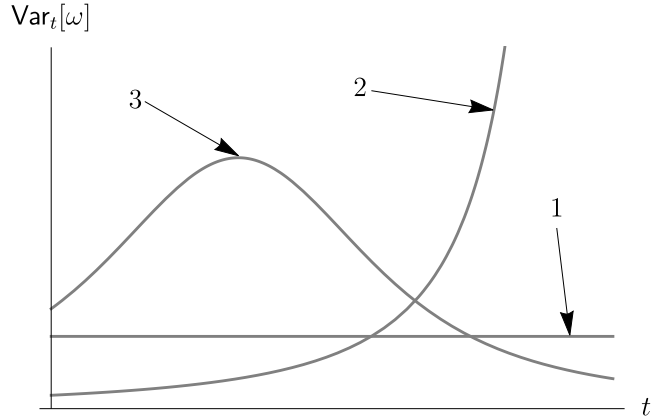}
\caption{Possible time dependent evolutions of the variance in the model \eqref{eq1:1}. $(1)$ is the initial normal distribution; $(2)$ is the initial exponential distribution, and $(3)$ the initial truncated on the interval $[0,5]$ exponential distribution.}\label{fig:2}
\end{figure}

The last observation becomes especially important in many models of \textit{quantitative genetics}, which often take the form \eqref{eq1:5}, and in many of which it is assumed that the variance is fixed (e.g., \cite{duffy2007rapid}). It should be kept in mind that, according to Proposition \ref{pr:2:1}, the only situation that is being modeled in this case is the initial normal distribution.

\section{Gypsy moth problem}\label{sec:4}
As an another example of the utility of the key theoretical relation \eqref{eq1:7} consider a general mathematical model that accounts for host heterogeneity among individual gypsy moth larva with respect to susceptibility to the virus.
This model was treated in a series of papers by Greg Dwyer and his co-authors \cite{Dwyer2002,Dwyer2000,Dwyer1997,elderd2008host}. These papers combined accurate mathematical modeling, laboratory dose-response experiments, field transmission experiments, and observations of naturally occurring populations. In the initial works it was concluded that ``a model incorporating host heterogeneity in susceptibility to the virus gives a much better fit to data on virus dynamics [...] than does a classical model''\cite{Dwyer1997} and ``our experimental estimates of virus transmission rates and levels of heterogeneity in susceptibility in gypsy moth populations give model dynamics that closely approximate the dynamics of real gypsy moth populations''\cite{Dwyer2000}. However, in a more recent work \cite{elderd2008host} the original model was replaced with an alternative one, because it was found that ``Our data show that heterogeneity in infection risk in this insect is so high that it leads to a stable equilibrium in the models, which is inconsistent with the outbreaks seen in North American gypsy moth populations''\cite{elderd2008host}. In more technical terms, it was observed that highly heterogeneous populations (for which the coefficient of variation $\textsf{CV}$ is bigger than one) exhibit oscillations with large amplitudes, whereas the mathematical model was able to produce such behaviors only for $\textsf{CV}<1$. In this section we show that the conclusion to refute the initial mathematical model was based on a small omission in the original analysis to ignore the key relation \eqref{eq1:7}.

\subsection{Mathematical model}
We mostly follow the notation from \cite{Dwyer2000} with several minor changes. Let $s(t,\omega)$ be the density of susceptible hosts having the susceptibility that is characterized by the parameter value $\omega$. Therefore, $s(0,\omega)$ defines the initial distribution of susceptibility before the disease starts, and $S(t)=\int_\Omega s(t,\omega)\D \omega$ is the total density of the host population at time $t$. Let $P(t)$ be the density of infectious cadavers at $t$, $\tau$ be the time between infection and death, and $\mu$ be the breakdown rate of the cadavers on the foliage. Then the mathematical model takes the form
\begin{equation}\label{eq:1}
\begin{split}
  \frac{\partial s}{\partial t}(t,\omega) & =-\omega s(t,\omega)P(t), \\\
  \frac{\D P}{\D t}(t)  & = P(t-\tau)\int_\Omega \omega s(t-\tau,\omega)\D \omega-\mu P(t),
\end{split}
\end{equation}
where $\Omega$ is the set of admissible values of $\omega$, e.g., $\Omega=[0,\infty)$. The initial conditions are
\begin{equation}\label{eq:2}
\begin{split}
     s(0,\omega)&=s_0(\omega)=S_0p_0(\omega)=S_0p(0,\omega),\\
     P(0)&=P_0,
\end{split}
\end{equation}
where $S_0=S(0)$ is the total initial density of the host population, and $p_0(\omega)$ is the \textit{initial} distribution of the parameter $\omega$ in the host population (such that $\int_\Omega p_0(\omega)d\omega=1$ and $p_0(\omega)\geq 0$ when $\omega\in\Omega$). Note that model \eqref{eq:1}, \eqref{eq:2} is a straightforward generalization of model \eqref{eq1:5}, \eqref{eq1:2}.

Inasmuch as of the most interest is the final epidemic size $x$ (i.e., the proportion of the population that gets infected during the epidemics), it is possible to allow the time to go to infinity to obtain a transcendental equation for $x$ in terms of the initial conditions and model parameters. Instead of this approach, that was used by Dwyer et al, we will show some intermediate steps.

First, integration of the first equation in \eqref{eq:1} with respect to $\omega$ implies
\begin{equation}\label{eq:2a}
\begin{split}
  \frac{\D S}{\D t}(t) & =-\textsf{E}_t[\omega] S(t)P(t), \\
  \frac{\D P}{\D t}(t)  & =\textsf{E}_{t-\tau}[\omega]S(t-\tau)P(t-\tau)-\mu P(t),
\end{split}
\end{equation}
where, as before,
$$
\textsf{E}_t[\omega]=\int_\Omega \omega p(t,\omega)\D\omega,
$$
is the current \emph{mean} of the distribution of $\omega$ in the host population at the moment $t$. It is not constant, but a function of time (intuitively, it must decrease, since the infection washes out first those who have initially higher values of $\omega$). As we discussed in Section 1, we have
\begin{equation}\label{eq:4}
\textsf{E}_t[\omega]=\frac{\D}{\D\lambda} \log \textsf{M}(0,\lambda)|_{\lambda=q(t)}\,,
\end{equation}
where $q(t)$ solves
\begin{equation}\label{eq:5}
\frac{\D q}{\D t}(t)=-P(t),
\end{equation}
with the initial condition $q(0)=0$. Therefore, instead of two equations in \eqref{eq:1} we end up with three ODE \eqref{eq:2a}, \eqref{eq:4}, \eqref{eq:5}, which are equivalent to formally infinite dimensional system \eqref{eq:1}. Furthermore, it can be shown that \eqref{eq:2a}, \eqref{eq:4}, \eqref{eq:5} are equivalent to
\begin{equation}\label{eq:6}
\begin{split}
  \frac{\D S}{\D t}(t) & =-h(S(t))P(t), \\
  \frac{\D P}{\D t}(t)  & =h(S(t-\tau))P(t-\tau)-\mu P(t),
\end{split}
\end{equation}
where function $h(S)$ is given by ($\textsf{M}^{-1}$ is the inverse function to $\textsf{M}$)
$$
h(S)=S_0\frac{\D}{\D\lambda}\textsf{M}^{-1}(0,\lambda)|_{\lambda=S(t)/S_0}\,.
$$
More details can be found in \cite{Novozhilov2008}. We remark that no special assumption was made so far about the initial susceptibility distribution.

Now assume that the initial distribution is a gamma distribution with parameters $\nu$ and $k$. We can find explicitly that for the initial gamma distribution system \eqref{eq:1} takes the form
\begin{equation}\label{eq:7}
\begin{split}
  \frac{\D S}{\D t}(t) & =-\frac{k}{\nu}\left(\frac{S(t)}{S_0}\right)^{1/k}S(t)P(t), \\
  \frac{\D P}{\D t}(t)  & =\frac{k}{\nu}\left(\frac{S(t-\tau)}{S_0}\right)^{1/k}S(t-\tau)P(t-\tau)-\mu P(t),
\end{split}
\end{equation}
which coincides with equations (3)-(4) in \cite{elderd2008host}, obtained by a different means.

In \cite{Dwyer2000} it was stated that ``Our results, however, can also be derived without this assumption'', where by ``assumption'' the initial gamma distributions is meant. In particular, what was used ``instead'' of the initial assumption of the gamma distribution is an ``alternative'' assumption on the conservation of the coefficient of variation.  Proposition \ref{pr:2:1} shows these two assumptions are equivalent and hence, not surprisingly, lead to the same system of equations. Therefore, instead of a general situation, a very particular mathematical model was analyzed.

\subsection{Toy example}
Here we give an example of a population with the \textit{initial} coefficient of variation $\textsf{CV}>1$, which exhibits stable oscillations.

For the following we will need the equation for the final epidemic size $x$
$$
1-x=\textsf{M}\left(0,-\frac{S_0x+P_0}{\mu}\right),
$$
which can be obtained directly from \eqref{eq:7}.

Consider a general family of distributions, defined by its moment generating function
$$
\textsf{M}(0,\lambda)=\exp\left[-\rho\left(1-\left(\frac{\alpha}{\alpha-\lambda}\right)^\nu\right)\right]
$$
with parameters $\alpha>0,\,\nu>-1,\nu\rho>0$. This is the co-called variance function distributions \cite{aalen2008survival}. We find, using \eqref{eq1:7}, that
$$
\textsf{E}_t[\omega]=\rho\left(\frac{\alpha}{\alpha-q(t)}\right)^\nu\frac{\nu}{\alpha-q(t)}\,,\quad \textsf{Var}_t[\omega]=\textsf{E}_t[\omega]\frac{\nu+1}{\alpha-q(t)}\,.
$$
Therefore, to guarantee that the coefficient of variation decreases with time, we should take $\nu<0,\,\rho<0$. For example, parameters $\rho = -3,\, \nu = -0.2,\alpha = 0.6$ imply that $\textsf{E}_0[\omega]=1,\,\textsf{CV}(0)=1.1547>1$, and both the mean and the coefficient of variation will decrease with time.

To produce population dynamics, we, following \cite{elderd2008host}, consider the system
\begin{align*}
N_{t+1}&=\lambda N_t(1-x_t),\\
Z_{t+1}&=\phi N_t x_t+\gamma Z_t,
\end{align*}
where $N_t$ and $Z_t$ are the initial host and pathogen densities in generation $t$, such that every time in the final epidemic size equation we take $S_0=N_t,\,P_0=Z_t$, $x_t$ is the final epidemic size for the year $t$, $\lambda,\phi,\gamma$ are the model specific parameters, which we take $\lambda=5.5,\phi=35,\gamma=0$ to coincide with the values used to produce Fig. 1 in \cite{elderd2008host}. 

The result of our simulations is given in Fig. \ref{fig:1}.
\begin{figure}[!ht]
\centering
\includegraphics[width=0.95\textwidth]{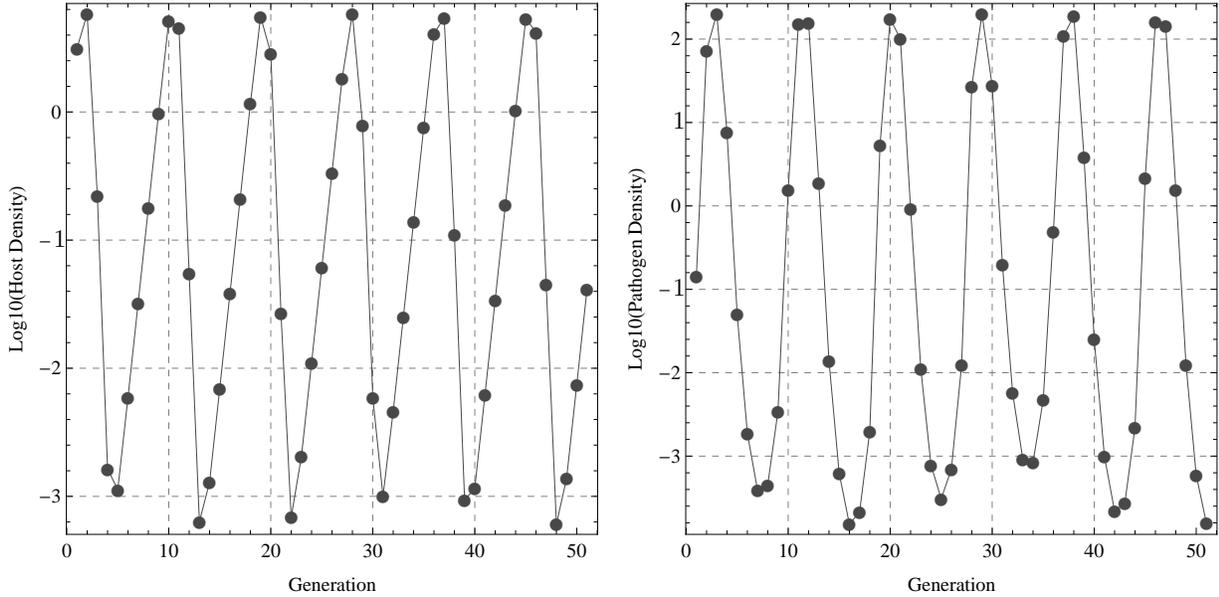}
\caption{Dynamics of the insect-pathogen models with the initial coefficient of variation exceeding 1. For the details and parameter values see the text.}\label{fig:1}
\end{figure}
We note that despite the fact that the initial coefficient of variation is above 1 (the population is highly heterogeneous), we still observe stable oscillations of large amplitude. This is in contrast with Fig. 1B in \cite{elderd2008host}, where a stable equilibrium for the model with $\textsf{CV}>1$ is shown. Therefore, in general, the conclusion from \cite{elderd2008host} that the model with the constant infection risk cannot produce the sustained oscillations if $\textsf{CV}>1$ should be replaced with the conclusion that the model with the constant infection risk \textit{and} the initial gamma distribution of susceptibility cannot produce the sustained oscillations if $\textsf{CV}>1$. Moreover, clear understanding of the implications of the key relation \eqref{eq1:7} allows us producing a mathematical model with the high initial coefficient of variation and stable oscillations.

\section{Concluding words}
In this text we discussed the evolution of distributions in special and yet very flexible mathematical models with parametric heterogeneity of the form \eqref{eq1:5} (and its possible direct generalizations, such as considered in Section \ref{sec:4} model with delay). We found that fixing the functional form of our models implies quite severe restrictions on the possible time-dependent evolution of the trait distributions under study. For instance a simple and quite reasonable for many situations assumption on the constant time independent variance turns out to be equivalent to the assumption on the initial normal distribution. Analogously, fixing the coefficient of variation is equivalent to fixing the initial (and for any future time moment) distribution as gamma distribution. The crucial tool in deriving such observations turned out to be the key relation \eqref{eq1:7}, which is also well known in the frailty theory \cite{aalen1994effects,aalen2008survival}. One of the important conclusions is that fixing the initial probability distribution can result in very different time dependent behavior of system's characteristics, as we illustrated in Section 3 by considering the evolution of the variance. Moreover, such an assumption, as we show in Section \ref{sec:4}, may also lead to important omissions in the model analysis and hence incorrect conclusions.

More importantly, however, is that a clear understanding of the mutual relation of the trait distributions and evolution of parameters (such as mean, variance, etc) allows devising efficient statistical procedures to test the data against our mathematical models. The statistical tests for the heterogeneous models of the form \eqref{eq1:5} only start to appear (one theoretical paper is \cite{tsachev2017set}). It is clear, however, that the analysis presented in this note implies that for a reliable conclusion the data on the trait distributions must be collected at several (two or more) time moments. Only in this case instead of testing usually quite approximate data against an exact hypothesis of parameter distribution (recall from Section 3 that we do need the exact distribution the determine the dynamics, no finite number of moments is enough), one can test whether, e.g., the variance (or coefficient of variation) does not change with time, and the analytical results collected in the present paper will make the conclusion on the exact form of the underlying distribution automatic.

\section{Appendix}
In Appendix we collect a few more examples of specific ODE that yield as their solutions various mgf. 

Let
$$
\mu_n(q)=\textsf{E}_q[(w-\textsf{E}_q[w])^n]
$$
be the $n$-th central moment of a distribution. Then for each prescribed $\mu_n(q)$ equality \eqref{eq1:7} yields a differential equation. For $n=2$ this equation is given by \eqref{eq2:1}. Consider the case $n=3$. Then, using the notation introduced in Section 2, we have
$$
\frac{m'''(q)}{m(q)}-3\frac{m''(q)}{m(q)}\frac{m'(q)}{m(q)}+3\left(\frac{m'(q)}{m(q)}\right)^3=\mu_3(q)\,.
$$
If we rewrite this equation in terms of the variable
$$
z(q)=\frac{m'(q)}{m(q)}\,,
$$
then it takes a very simple form
$$
z''(q)=\mu_3(q),
$$
which can be easily integrated for each particular $\mu_3(q)$. For instance, if $\mu_3(q)=C$ is a constant in internal time $q$, and using the natural initial conditions
$$
m(0)=1,\quad m'(0)=\mu,\quad m''(0)=\sigma^2+\mu^2,
$$
we obtain
$$
m(q)=\exp\left(\frac{Cq^3}{6}+\frac{\sigma^2 q^2}{2}+\mu q\right),
$$
which gives exactly the mgf of the normal distribution for $C=0$, as it is shown in a different way in Section 3.

For the fourth and fifth central moments the ODE in terms of $z$ takes the form, respectively,
\begin{align*}
z'''+3(z')^2&=\mu_4,\\
z^{(4)}+10 z'' z'&=\mu_5,
\end{align*}
whose general solutions can be written in term of the Weierstrass zeta functions.

\paragraph{Acknowledgements:} ASN would like to thank Bret Elderd for a profitable discussion while preparing Section \ref{sec:4} of this text. GPK's research was partially supported by the Intramural Research Program of the NCBI, NIH.


\end{document}